%% file: cav14.tex
\documentclass{llncs}
\pdfoutput=1
\input{macros.tex}

\begin{document}
\title{Synthesizing Finite-state Protocols from Scenarios and Requirements\thanks{This is the working draft of a paper currently in submission. (February 10, 2014)}}
\author{Rajeev Alur\inst{1} \and
    Milo Martin\inst{1} \and
    Mukund Raghothaman\inst{1} \and
    Christos Stergiou\inst{1,2} \and
    Stavros Tripakis\inst{2,3} \and
    Abhishek Udupa\inst{1}}

\institute{University of Pennsylvania \and
   University of California, Berkeley \and Aalto University}

\maketitle

\begin{abstract}
\input{abstract.tex}
\end{abstract}
\input{intro.tex}
\input{methodology.tex}
\input{formalization.tex}

\input{solutions.tex}

\input{case_studies.tex}

\input{conclusions.tex}
\bibliographystyle{splncs}
\bibliography{cav14}
\newpage
\input{appendix.tex}

\end{document}

%% file: macros.tex
\usepackage{amsmath}
\usepackage{amsfonts}
\usepackage{color}
\usepackage{etoolbox}
\usepackage{mdwtab}
\usepackage{microtype}
\usepackage{multirow}
\usepackage{pifont}
\usepackage{tikz}

\usetikzlibrary{arrows}
\usetikzlibrary{automata}
\usetikzlibrary{calc}
\usetikzlibrary{decorations.pathmorphing}
\usetikzlibrary{positioning}
\usetikzlibrary{shapes}
\usetikzlibrary{shapes.geometric}
\usetikzlibrary{shapes.misc}
\usetikzlibrary{decorations}

\usetikzlibrary{automata, arrows, positioning}
\tikzset{>=stealth', state/.append style={minimum size=.4cm}}

\newcommand{\tr}[1]{\stackrel{#1}{\rightarrow}}
\newcommand{\defeq}{\;\widehat{=}\;}
\newcommand{\st}{\mathrel{|}}
\newcommand{\true}{\mathit{true}}
\newcommand{\false}{\mathit{false}}

\newcommand{\tle}{\ensuremath{\mathsf{E}}}
\newcommand{\tla}{\ensuremath{\mathsf{A}}}
\newcommand{\tlf}{\ensuremath{\mathsf{F}}}
\newcommand{\tlg}{\ensuremath{\mathsf{G}}}

%% file: abstract.tex
Scenarios, or Message Sequence Charts, offer an intuitive way of
describing the desired behaviors of a distributed protocol.  In this
paper we propose a new way of specifying finite-state protocols using
scenarios: we show that it is possible to automatically derive a
distributed implementation from a set of scenarios augmented with a
set of safety and liveness requirements, provided the given scenarios
adequately \emph{cover} all the states of the desired
implementation. We first derive incomplete state machines from the
given scenarios, and then synthesis corresponds to completing the
transition relation of individual processes so that the global product
meets the specified requirements. This completion problem, in general,
has the same complexity, PSPACE, as the verification problem, but
unlike the verification problem, is NP-complete for a constant number
of processes. We present two algorithms for solving the completion
problem, one based on a heuristic search in the space of possible
completions and one based on OBDD-based symbolic fixpoint computation.
We evaluate the proposed methodology for protocol specification and
the effectiveness of the synthesis algorithms using the classical
alternating-bit protocol.

%% file: intro.tex
\section{Introduction}






In formal verification, a system model is checked against correctness
requirements to find bugs.  Sustained research in improving
verification tools over the last few decades has resulted in powerful
heuristics for coping with the computational intractability of
problems such as Boolean satisfiability and search through the
state-space of concurrent processes.  The advances in these analysis
tools now offer an opportunity to develop new methodologies for system
design that allow a programmer to specify a system in more intuitive
ways.  In this paper, we focus on distributed protocols: the multitude
of behaviors arising due to asynchronous concurrency makes the design
of such protocols difficult, and the benefits of using model checkers
to debug such protocols have been clearly demonstrated.  Traditionally
a distributed protocol is described using communicating finite-state
machines (FSMs), and the goal of this paper is to develop a
methodology aimed at simplifying the task of specifying them.

A more intuitive way of specifying the desired behaviors of a protocol
is by \emph{scenarios}, where each scenario describes an expected
sequence of message exchanges among participating processes. Such
scenarios are used in textbooks and classrooms to explain the protocol
and can be specified using the intuitive visual notation of Message
Sequence Charts.  In fact, the MSC notation is standardized by
IEEE~\cite{IEEEMsc},
and it is supported by some system development environments as design
supplements. These observations raise the question: is it plausible to
ask the designer to provide enough scenarios so that the protocol
implementation can be automatically synthesized?  Although one cannot
expect a designer to provide scenarios that include all the possible
behaviors, our key insight is that even a \emph{representative} set
of scenarios \emph{covers} all the states of the desired
implementation.  From a scenario, (local) states of a process are
obtained from explicit state-labels that appear as annotations as well
as from the histories of events in which the process participates. If
we consider all the states and the input/output transitions out of
these states for a given process that appear in the given set of
scenarios, we obtain a \emph{skeleton} of the desired FSM
implementation of that process. The synthesis problem now corresponds
to \emph{completing} this skeleton by adding transitions. This
requires the synthesizer to infer, for instance, how to process a
particular input event in a particular state even when this
information is missing from the specified scenarios.  The more such
completions that the synthesizer can learn successfully, the lower
the burden on the designer to specify details of each and every case.
To rule out incorrect completions, we ask the designer to provide a
model of the environment and correctness requirements. Some
requirements such as absence of deadlocks can be generic to all the
protocols, whereas other requirements specific to the coordination
problem being solved by the protocol are given as finite-state
monitors for safety and liveness properties commonly used in model
checkers.

The synthesis problem then maps to the following \emph{protocol
completion} problem: given (1) a set of FSMs with incomplete
transition functions, (2) a model of the environment, and (3) a set of
safety/liveness requirements, find a completion of the FSMs so that
the composition satisfies all the requirements.  We show this problem,
similar to the model checking problem, to be \textsc{Pspace}-complete,
but, unlike the model checking problem, to be \textsc{NP}-hard for
just one process.  We focus on two approaches to solving this problem:
the first performs a search through the space of possible completions
with heuristics guiding the search order and the second uses
OBDD-based symbolic model checking to compute the set of correct
completions by encoding the unknown targets of transitions as rigid
variables.

To evaluate our methodology, we consider the Alternating Bit Protocol,
a classical solution to provide reliable transmission using unreliable
channels.  The canonical description of the protocol~\cite{kurose2009}
uses four scenarios to explain its behavior. It turns out that the
first scenario corresponding to the typical behavior contains a
representative of each local state of both the sender and receiver
processes. Our symbolic algorithm for protocol completion is able to
find the correct implementation from just one scenario, and thus,
automatically learn how to cope with message losses and message
duplications. We then vary the input, both in terms of the set of
scenarios and the set of correctness requirements, and study how it
affects the computational requirements and the ability to learn the
correct protocol for both the completion algorithms.

\subsection*{Related Work}

Our work builds on techniques and tools for model
checking~\cite{CGP00} and also on the rich literature for formal
modeling and verification of distributed protocols~\cite{Lynch}.

The problem of deriving finite-state implementations from formal
requirements specified, for instance, in temporal logic, is called
\emph{reactive synthesis}, and has been studied
extensively~\cite{RW89,PnueliRosner,BJPPS12}.  When the implementation
is required to be distributed, the problem is known to be
undecidable~\cite{PR90,TripakisIPL04,FS05,LamouchiThistle00}.  In \emph{bounded synthesis},
one fixes a bound on the number of states of the implementation, and
this allows algorithmic solutions to distributed
synthesis~\cite{FS13}.  Another approach uses \emph{genetic
programming}, combined with model checking, to search through protocol
implementations to find a correct one, and has been shown to be
effective in synthesizing protocols such as leader
election~\cite{KP08,KatzPeled09}.

Specifying a reactive system using example scenarios has also a long
tradition.  In particular, the problem of deriving an implementation
that exhibits at least the behaviors specified by a given set of
scenarios is well-studied (see, for instance,~\cite{AEY03,UchitelKramer03}).  A
particularly well-developed approach is \emph{behavioral
programming}~\cite{HMW12} that builds on the work on an extension of
message sequence charts, called \emph{live sequence
charts}~\cite{DammH01}, and has been shown to be effective for
specifying the behavior of a single controller reacting with its
environment.  More recently, scenarios --- in the form of ``flows''
--- have been used in the modular verification of cache coherence
protocols~\cite{Murali09}.

Our approach of using \emph{both} the scenarios and the requirements
in an integrated manner, and using scenarios to derive incomplete
state machines, offers a conceptually new methodology compared to the
existing work.  We are inspired by recent work on program
sketching~\cite{SRBE05} and on protocol
specification~\cite{Transit13}. Compared to
\textsc{Transit}~\cite{Transit13}, in this paper we limit ourselves to
finite-state protocols, but consider both safety and liveness
requirements, and provide a \emph{fully automatic} synthesis
procedure.

The protocol completion problem itself has conceptual similarities to
problems such as \emph{program repair} studied in the
literature~\cite{JGB05}, but differs in technical details.


%% file: methodology.tex
\section{Methodology}
\label{sec_methodology}



We explain our methodology by illustrating it on an example, the
well-known Alternating Bit Protocol (ABP). The ABP protocol ensures
reliable message transmission over unreliable channels which can
duplicate or lose messages. As input to the synthesis tool the user
provides the following:
\begin{itemize}
\item The \emph{protocol skeleton}: this is a set of processes which
are to be synthesized, and for each process, the \emph{interface} of
that process, i.e., its inputs and outputs.
\item The \emph{environment}: this is a set of processes which are
known and fixed, that is, are not to be synthesized nor modified in
any way by the synthesizer. The environment processes interact with
the protocol processes and the product of all these processes forms a
\emph{closed} system, which can be model-checked against a formal
specification.
\item A \emph{specification}: this is a set of formal requirements.
These can be expressed in different ways, e.g., as temporal logic
formulas, safety or liveness (e.g., B\"uchi) monitors, or
``hardwired'' properties such as absence of deadlock.
\item A set of \emph{scenarios}: these are example behaviors of the
system.  In our framework, a scenario is a type of \emph{message
sequence chart} (MSC).
\end{itemize}

In the case of the ABP example, the above inputs are as follows.  The
overall system is shown in Figure~\ref{fig_abp}.  The protocol
skeleton consists of the two unknown processes \emph{ABP Sender} and
\emph{ABP Receiver}. Their interfaces are shown in the figure, e.g.,
ABP Sender has inputs $a_0'$, $a_1'$, and $timeout$ and outputs
$send$, $p_0$, and $p_1$.  The environment processes are: \emph{Forward
Channel} (FC) (from ABP Sender to ABP Receiver, duplicating and
lossy), \emph{Backward Channel} (BC) (from ABP Receiver to ABP Sender,
also duplicating and lossy), \emph{Timer} (sends $timeout$ messages to
ABP Sender), \emph{Safety Monitor}, and a set of \emph{Liveness
Monitors}.

As specification for ABP we will use the following requirements: (1)
deadlock-freedom, i.e., absence of reachable global deadlock states
(in the product system); (2) safety, captured by a safety monitor such
as the one depicted in Figure~\ref{fig_abp}; (3) B\"uchi liveness
monitors, that accept incorrect infinite executions in which either a
send message is not followed by a deliver, a deliver is not followed
by a send, or a send never appears, provided that the channels are
fair; as well as (4) a property that we call \emph{non-blockingness},
which informally requires that in every reachable global state, if a
process wants to send a message to another process, then the latter
must be able to receive it.  Non-blockingness allows to specify that
the system does not have local deadlocks, where a process is blocked
from making progress while the system as a whole is not deadlocked.

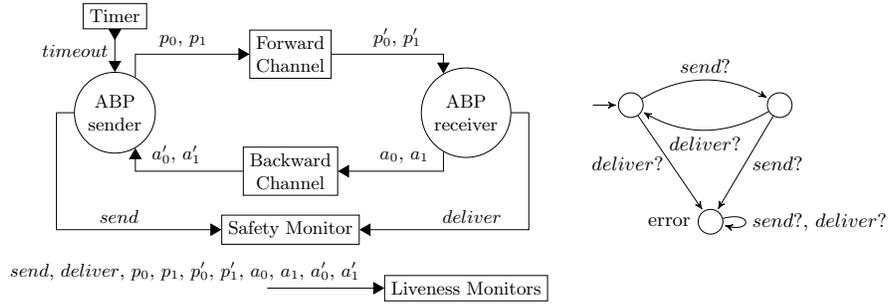
\begin{figure}[!t]
  \centering
  \begin{minipage}[t]{0.6\textwidth}
  \resizebox{1.0\textwidth}{!}{
    \begin{tikzpicture}
      [every node/.style={on grid, draw, align=center}, node distance=2.5cm,
      every edge/.style={>->,draw,>=triangle 60}]
      \node[circle] (abpsender) {ABP\\sender};
      \node[above right=1cm and 3cm of abpsender] (forwardchannel) {Forward\\Channel};
      \node[below right=1cm and 3cm of abpsender] (backwardchannel) {Backward\\Channel};
      \node[circle,below right=1cm and 3cm of forwardchannel] (abpreceiver) {ABP\\receiver};
      \node[below=1cm of backwardchannel] (safetymonitor) {Safety Monitor};
      \node[below right=2cm and 3cm of backwardchannel] (livenessmonitor) {Liveness Monitors};
      \node[above=1.65cm of abpsender] (timer) {Timer};
      \path (abpreceiver) edge[>=] ($(abpreceiver.east) + (0.3cm, 0)$)
      (timer) edge node[draw=none,left] {$timeout$} (abpsender)
      ;
      \tikzstyle{line} = [draw, -triangle 60]
      \begin{scope}[every node/.style={draw=none}]
        \path[line] ($(livenessmonitor.west) - (2cm, 0)$) -- (livenessmonitor);
        \node [above left=-.2cm and .3cm of livenessmonitor] {$send$, $deliver$, $p_0$, $p_1$, $p_0'$, $p_1'$, $a_0$, $a_1$, $a_0'$, $a_1'$};
        \path[line] (abpsender.60) |- (forwardchannel);
        \node[above left=-0.4cm and .6cm of forwardchannel] {$p_0$, $p_1$};
        \node[above right=-0.4cm and .6cm of forwardchannel] {$p_0'$, $p_1'$};
        \path[line] (forwardchannel) -| (abpreceiver.120);
        \path[line] (backwardchannel) -| (abpsender.-60);
        \path[line] (abpreceiver.240) |- (backwardchannel);
        \node[above left=-0.4cm and .6cm of backwardchannel] {$a_0'$, $a_1'$};
        \node[above right=-0.4cm and .6cm of backwardchannel] {$a_0$, $a_1$};
        \path[line] ($(abpsender.west) + (-.3cm,0)$) |- (safetymonitor);
        \path[draw] (abpsender) -- ($(abpsender.west) + (-.3cm,0)$);
        \path[line] ($(abpreceiver.east) + (0.3cm,0)$) |- (safetymonitor);
        \node[above right=-.25cm and 1.3cm of safetymonitor] {$deliver$};
        \node[above left=-.25cm and 1.3cm of safetymonitor] {$send$};
      \end{scope}
    \end{tikzpicture}}
  \end{minipage}\quad
  \begin{minipage}{0.35\textwidth}
    \centering
    \vspace{-40mm}
    \resizebox{1.0\textwidth}{!}{
      \begin{tikzpicture}[node distance=2cm,->,>=stealth',shorten >=1pt,align=center]
        \path
        node[state, initial, initial text={}] (s0) {}
        node[state, right=of s0] (s1) {}
        node[state, below right=1.6cm and 1.0cm of s0, label=left:{error}] (s2) {}
        (s0) edge[bend left] node[above] {$send?$} (s1)
        (s1) edge[bend left] node[below] {$deliver?$} (s0)
        (s0) edge node[left] {$deliver?$} (s2)
        (s1) edge node[right] {$send?$} (s2)
        (s2) edge[loop right] node[right] {$send?$, $deliver?$} (s2)
        ;
      \end{tikzpicture}}
  \end{minipage}
  \caption{ABP system architecture (left) and the safety monitor which
  ensures that send and deliver messages alternate (right)}
  \label{fig_abp}
\end{figure}


We will use the four message sequence charts
shown in Figure~\ref{fig_abp_scenarios} to describe the behavior of
the ABP protocol. They come from a textbook on computer
networking~\cite{kurose2009}. The first scenario describes the
behavior of the protocol when no packets or acknowledgments are lost
or duplicated. The second and the third scenarios correspond to packet
and acknowledgment loss respectively. Finally, the fourth scenario
describes the behavior of ABP on premature timeouts and/or packet
duplication.

\input{abp_scenarios}

A candidate solution to the ABP synthesis problem is a pair of
processes, one for the ABP Sender and one for the ABP Receiver. Such a
candidate is a valid solution if: (a) the two processes respect their
I/O interface and satisfy some additional requirements such as determinism
(these are defined formally in Section~\ref{subsec_fs_io_automata}),
(b) the overall ABP system (product of all processes) may exhibit each
of the input scenarios, and (c) it satisfies the correctness
requirements.

The output of the BDD-based algorithm when run with the requirements
mentioned above and only the first scenario from
Figure~\ref{fig_abp_scenarios} is shown in
Figure~\ref{fig_abp_computed_solution}.  It can be checked that these
solutions are ``similar/equivalent'' in the sense that they satisfy
the same intuitive properties that one expects from the ABP protocol.
In particular, the computed solution in
Figure~\ref{fig_abp_computed_solution} eagerly retransmits the
appropriate packet when an unexpected acknowledgment is
received. This
behavior might incur additional traffic but satisfies all the safety
and liveness properties for the ABP protocol. The computed solution
for the ABP receiver is the same as the manually constructed automaton
shown in Figure~\ref{fig_abp_ideal_solution}.

\input{abp_ideal}
\input{abp_computed}

%% file: abp_scenarios.tex
\newcommand\msg[5]{
  \ifnumequal{#1}{0}{\def\source{\sender}}{\def\source{\receiver}}
  \ifnumequal{#1}{0}{\def\target{\receiver}}{\def\target{\sender}}
  \ifnumequal{#1}{0}{\def\sourcepos{left}}{\def\sourcepos{right}}
  \ifnumequal{#1}{0}{\def\targetpos{right}}{\def\targetpos{left}}
  \path (\source,#2) edge[message] node[messagenode] {#3} (\target,#2-.5);
  \path (\source,#2) node[align=right,\sourcepos=.2cm] {#4};
  \path (\target,#2-.5) node[\targetpos=.2cm,align=left] {#5};
}

\tikzstyle{lane} = [draw, loosely dotted, ->, >=latex, line width=1.3]
\tikzstyle{message} = [draw, ->, >=latex, line width=2]
\tikzstyle{messagenode} = [sloped, above, pos=.3]

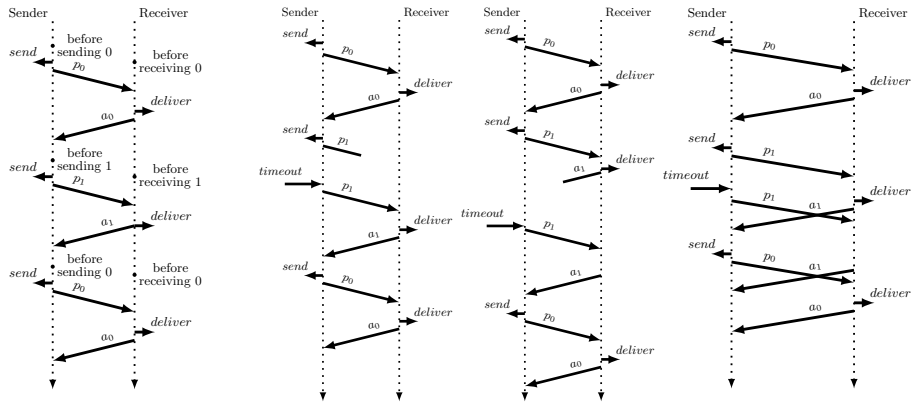
\begin{figure}[!t]
  \centering
  \begin{minipage}[t]{.22\textwidth}
    \vspace{0pt}
    \hspace{-40pt}
    \centering
    \resizebox{1.0\textwidth}{!}{\begin{tikzpicture}
      \def\lanelen{9}
      \def\sender{0}
      \def\receiver{2}
      \path
      (0,-.6) node[circle, fill, inner sep=1pt, minimum size=1pt] {}
      (0,-.6) node[anchor=west, align=center] {before\\[-2pt]sending 0}
      ;
      \path
      (0,-1) edge[message] node[above left] {$send$} (-.5,-1)
      (\sender,0) node[above left] {Sender}
      (\sender,0) edge[lane] (\sender,-\lanelen)
      (\receiver,0) node[above right] {Receiver}
      (\receiver,0) edge[lane] (\receiver,-\lanelen)
      (2,-1) node[circle, fill, inner sep=1pt, minimum size=1pt] {}
      (2,-1) node[anchor=west, align=center] {before\\[-2pt]receiving 0}
      (2,-6.2) node[circle, fill, inner sep=1pt, minimum size=1pt] {}
      (2,-6.2) node[anchor=west, align=center] {before\\[-2pt]receiving 0}
      ;
      \msg{0}{-1.2}{$p_0$}{}{}
      \path (2,-2.2) edge[message] node[above right] {$deliver$} (2.5,-2.2);
      \msg{1}{-2.4}{$a_0$}{}{}

      \path
      (0,-3.4) node[circle, fill, inner sep=1pt, minimum size=1pt] {}
      (0,-3.4) node[anchor=west, align=center] {before\\[-2pt]sending 1};

      \path (0,-3.8) edge[message] node[above left] {$send$}
      (-.5,-3.8)
      (2,-3.8) node[circle, fill, inner sep=1pt, minimum size=1pt] {}
      (2,-3.8) node[anchor=west, align=center] {before\\[-2pt]receiving 1}
      ;
      \msg{0}{-4}{$p_1$}{}{}
      \path (2,-5) edge[message] node[above right] {$deliver$} (2.5,-5);
      \msg{1}{-5}{$a_1$}{}{}
      \path (0,-6.4) edge[message] node[above left] {$send$} (-.5,-6.4);
      \msg{0}{-6.6}{$p_0$}{}{}
      \path (2,-7.6) edge[message] node[above right] {$deliver$} (2.5,-7.6);
      \msg{1}{-7.8}{$a_0$}{}{}
      \path
      (0,-6) node[circle, fill, inner sep=1pt, minimum size=1pt] {}
      (0,-6) node[anchor=west, align=center] {before\\[-2pt]sending 0}
      ;

    \end{tikzpicture}}
  \end{minipage}%
  \begin{minipage}[t]{.22\textwidth}
    \vspace{0pt}
    \centering
    \resizebox{1.0\textwidth}{!}{\begin{tikzpicture}
      \def\lanelen{10}
      \def\sender{0}
      \def\receiver{2}
      \path
      (0,-.6) edge[message] node[above left] {$send$} (-.5,-.6)
      (\sender,0) node[above left] {Sender}
      (\sender,0) edge[lane] (\sender,-\lanelen)
      (\receiver,0) node[above right] {Receiver}
      (\receiver,0) edge[lane] (\receiver,-\lanelen)
      ;
      \msg{0}{-.9}{$p_0$}{}{}
      \path (2,-1.9) edge[message] node[above right] {$deliver$} (2.5,-1.9);
      \msg{1}{-2.1}{$a_0$}{}{}

      \path (0,-3.1) edge[message] node[above left] {$send$} (-.5,-3.1);

      \path (\sender,-3.3) edge[message,-] node[messagenode,pos=.5] {$p_1$} (1,-3.3-.25);

      \path (-1,-4.3) edge[message] node[above left=.1cm and -.1cm] {$timeout$} (0,-4.3);
      \msg{0}{-4.5}{$p_1$}{}{}
      \path (2,-5.5) edge[message] node[above right] {$deliver$} (2.5,-5.5);
      \msg{1}{-5.7}{$a_1$}{}{}
      \path (0,-6.7) edge[message] node[above left] {$send$} (-.5,-6.7);
      \msg{0}{-6.9}{$p_0$}{}{}
      \path (2,-7.9) edge[message] node[above right] {$deliver$} (2.5,-7.9);
      \msg{1}{-8.1}{$a_0$}{}{}
    \end{tikzpicture}}
  \end{minipage}%
  \begin{minipage}[t]{.22\textwidth}
    \vspace{0pt}
    \resizebox{1.0\textwidth}{!}{\begin{tikzpicture}
      \def\lanelen{10}
      \def\sender{0}
      \def\receiver{2}
      \path
      (\sender,0) node[above left] {Sender}
      (\sender,0) edge[lane] (\sender,-\lanelen)
      (\receiver,0) node[above right] {Receiver}
      (\receiver,0) edge[lane] (\receiver,-\lanelen)
      ;
      \path (0,-.5) edge[message] node[above left] {$send$} (-.5,-.5);
      \msg{0}{-.7}{$p_0$}{}{}
      \path (2,-1.7) edge[message] node[above right] {$deliver$} (2.5,-1.7);
      \msg{1}{-1.9}{$a_0$}{}{}
      \path (0,-2.9) edge[message] node[above left] {$send$} (-.5,-2.9);
      \msg{0}{-3.1}{$p_1$}{}{}
      \path (\receiver,-4) node[right=.2cm] {};
      \path (\receiver,-4) edge[message,-] node[messagenode,pos=.5] {$a_1$} (1,-4-.25);
      \path (-1,-5.4) edge[message] node[above left=.1cm and -.1cm] {$timeout$} (0,-5.4);
      \msg{0}{-5.5}{$p_1$}{}{}
      \path (2,-3.9) edge[message] node[above right] {$deliver$} (2.5,-3.9);
      \msg{1}{-6.7}{$a_1$}{}{}
      \path (0,-7.7) edge[message] node[above left] {$send$} (-.5,-7.7);
      \msg{0}{-7.9}{$p_0$}{}{}
      \path (2,-8.9) edge[message] node[above right] {$deliver$} (2.5,-8.9);
      \msg{1}{-9.1}{$a_0$}{}{}
    \end{tikzpicture}}
  \end{minipage}%
  \begin{minipage}[t]{.28\textwidth}
    \vspace{0pt}
    \centering
    \resizebox{\textwidth}{!}{\begin{tikzpicture}
      \def\lanelen{9}
      \def\sender{0}
      \def\receiver{3}
      \path
      (\sender,0) node[above left] {Sender}
      (\sender,0) edge[lane] (\sender,-\lanelen)
      (\receiver,0) node[above right] {Receiver}
      (\receiver,0) edge[lane] (\receiver,-\lanelen)
      ;
      \path (0,-.5) edge[message] node[above left] {$send$} (-.5,-.5);
      \msg{0}{-.7}{$p_0$}{}{}
      \path (3,-1.7) edge[message] node[above right] {$deliver$} (3.5,-1.7);
      \msg{1}{-1.9}{$a_0$}{}{}
      \path (0,-3.1) edge[message] node[above left] {$send$} (-.5,-3.1);
      \msg{0}{-3.3}{$p_1$}{}{}
      \path (3,-4.4) edge[message] node[above right] {$deliver$} (3.5,-4.4);
      \msg{1}{-4.6}{$a_1$}{}{}
      \path (0,-5.7) edge[message] node[above left] {$send$} (-.5,-5.7);
      \msg{0}{-5.9}{$p_0$}{}{}
      \path (3,-6.9) edge[message] node[above right] {$deliver$} (3.5,-6.9);
      \msg{1}{-7.1}{$a_0$}{}{}

      \path (-1,-4.1) edge[message] node[above left=.1cm and -.1cm] {$timeout$} (0,-4.1);
      \msg{0}{-4.4}{$p_1$}{}{}
      \msg{1}{-6.1}{$a_1$}{}{}
    \end{tikzpicture}}
  \end{minipage}
  \caption{\label{fig_abp_scenarios}
Four scenarios for the alternating-bit protocol. 
From left to right:
No loss,
Lost packet,
Lost ACK,
Premature timeout/duplication.}
\end{figure}

%% file: abp_ideal.tex
\begin{figure}[!t]
  \centering
  \begin{minipage}{.4\textwidth}
    \centering
    \resizebox{\textwidth}{!}{
    \begin{tikzpicture}[node distance=2cm,->,>=stealth',shorten >=1pt, align=center]
      \node[state,initial,initial text={}] (s0) {};
      \node[state] (s1) [above right=.7cm and 1cm of s0] {};
      \node[state] (s2) [right=of s1] {};
      \node[state] (s3) [below right=.7cm and 1cm of s2] {};
      \node[state] (s4) [below left=.7cm and 1cm of s3] {};
      \node[state] (s5) [left=of s4] {};
      \path
      (s0) edge node[above left] {$send!$} (s1)
      (s1) edge[bend left] node[above] {$p_0!$} (s2)
      (s2) edge[bend left] node[below] {$timeout?$} (s1)
      (s2) edge[loop above] node[above] {$a_1'?$} (s2)
      (s2) edge node[above right] {$a_0'?$} (s3)
      (s3) edge node[below right] {$send!$} (s4)
      (s4) edge[bend right] node[above] {$p_1!$} (s5)
      (s5) edge[bend right] node[below] {$timeout?$} (s4)
      (s5) edge[loop above] node[above] {$a_0'?$} (s5)
      (s5) edge node[below left] {$a_1'?$} (s0)
      ;
    \end{tikzpicture}}
  \end{minipage}\qquad
  \begin{minipage}{.35\textwidth}
    \centering
    \resizebox{\textwidth}{!}{
    \begin{tikzpicture}[node distance=3cm,->,>=stealth',shorten >=1pt,align=center]
      \node[state,initial,initial text={}] (r0) {};
      \node[state, above right=.7cm and 1cm of r0] (r1) {};
      \node[state, right=1.4cm of r1] (r2) {};
      \node[state, below right=.7cm and 1cm of r2] (r3) {};
      \node[state, below left=.7cm and 1cm of r3] (r4) {};
      \node[state, left=1.4cm of r4] (r5) {};
      \path
      (r0) edge node[above left] {$p_0'?$} (r1)
      (r1) edge node[above] {$deliver!$} (r2)
      (r2) edge[bend left] node[above right] {$a_0!$} (r3)
      (r3) edge[bend left] node[below left] {$p_0'?$} (r2)
      (r3) edge node[below right] {$p_1'?$} (r4)
      (r4) edge node[above] {$deliver!$} (r5)
      (r5) edge[bend right] node[above right] {$a_1!$} (r0)
      (r0) edge[bend right] node[below left] {$p_1'?$} (r5)
      ;
    \end{tikzpicture}}
  \end{minipage}
  \caption{\label{fig_abp_ideal_solution} ABP ``manual'' solution: ABP Sender (left), ABP Receiver (right).}
\end{figure}
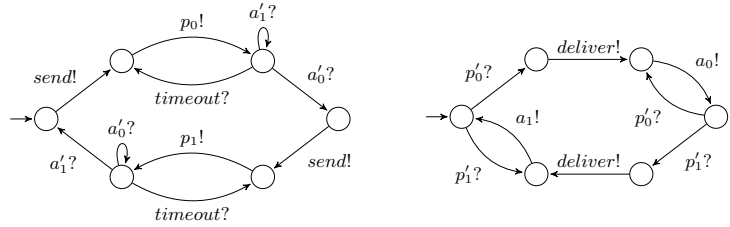

%% file: abp_computed.tex
\begin{figure}[ht]
  \centering
  \resizebox{0.4\textwidth}{!}{
    \begin{tikzpicture}[node distance=2cm,->,>=stealth',shorten >=1pt, align=center]
      \node[state,initial,initial text={}] (s0) {};
      \node[state] (s1) [above right=.7cm and 1cm of s0] {};
      \node[state] (s2) [right=of s1] {};
      \node[state] (s3) [below right=.7cm and 1cm of s2] {};
      \node[state] (s4) [below left=.7cm and 1cm of s3] {};
      \node[state] (s5) [left=of s4] {};
      \path
      (s0) edge node[above left] {$send!$} (s1)
      (s1) edge[out=70,in=110,looseness=0.7] node[above] {$p_0!$} (s2)
      (s2) edge[bend left] node[below] {$timeout?$} (s1)
      (s2) edge node[above] {$a_1'?$} (s1)
      (s2) edge node[above right] {$a_0'?$} (s3)
      (s3) edge node[below right] {$send!$} (s4)
      (s4) edge[out=110,in=70,looseness=0.7] node[above] {$p_1!$} (s5)
      (s5) edge[bend right] node[below] {$timeout?$} (s4)
      (s5) edge node[above] {$a_0'?$} (s4)
      (s5) edge node[below left] {$a_1'?$} (s0)
      ;
    \end{tikzpicture}}
  \caption{\label{fig_abp_computed_solution} Solution computed for ABP
    Sender by the BDD-based symbolic algorithm using only the first
    scenario.}
\end{figure}
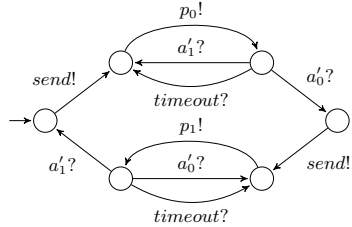

%% file: formalization.tex
\section{The Automata Completion Problem}
\label{sec_formalization}
We now describe how the problem, which we have set up in
Section~\ref{sec_methodology}, can be viewed as a problem of
completing the transition relations of finite IO automata.


\subsection{Finite-state Input-Output Automata}
\label{subsec_fs_io_automata}

A {\em finite-state input-output automaton} is a tuple $A =
(Q,q_0,I,O,T)$ where $Q$ is a finite set of states, $q_0\in Q$ is the
initial state, $I$ is a finite (possibly empty) set of inputs, $O$ is
a finite (possibly empty) set of outputs, with $I \cap O=\emptyset$
and $T\subseteq Q\times (I\cup O) \times Q$ is a finite set of
transitions\footnote{The framework and synthesis algorithms can easily
  be extended to handle internal transitions as well, but we suppress
  this detail for simplicity of presentation.}.

We write a  transition $(q,x,q')\in T$ as $q\tr{x?}q'$  when $x\in I$,
and  as $q\tr{x!}q'$  when $x\in  O$.   We write  $q\tr{}q'$ if  there
exists $x$ such that $(q,x,q')\in T$.  A transition labeled with $x\in
I$  (respectively,  $x\in O$)  is  called  an {\em  input  transition}
(respectively, an {\em output transition}).

A state $q\in Q$ is called a {\em deadlock} if it has no outgoing transitions.
$q$ is called an {\em input state} if it
has at least one outgoing transition, and
all outgoing transitions from $q$ are input transitions.
$q$ is called an {\em output state} if it
has a single outgoing transition, which is an output transition.

Automaton $A$ is called {\em deterministic} if 
for every state $q\in Q$, if there are multiple outgoing
transitions from $q$, then all these transitions must be labeled with 
distinct inputs.
Determinism implies that every state $q\in Q$ is a deadlock, an
input state, or an output state.


Automaton $A$ is called {\em closed} if $I=\emptyset$.


A {\em safety monitor} is an automaton equipped with a set of {\em error
states} $Q_e$, $A = (Q,q_0,I,O,T,Q_e)$.
A {\em liveness monitor} is an automaton equipped with a set of {\em accepting
states} $Q_a$, $A = (Q,q_0,I,O,T,Q_a)$.
A monitor could be both safety and liveness, in which case it is a tuple
$A = (Q,q_0,I,O,T,Q_e,Q_a)$.

A {\em run} of an automaton $A$ is a finite or infinite sequence of transitions
starting from the initial state: $q_0 \tr{} q_1 \tr{} q_2 \tr{} \cdots$.
A state $q$ is called {\em reachable} if there exists a finite run reaching
that state: $q_0 \tr{} q_1 \tr{} \cdots \tr{} q$.
A safety automaton is called {\em safe} if it has no reachable error states.
An infinite run of a liveness automaton is called {\em accepting} if it visits
accepting states infinitely often.
A liveness automaton is called {\em empty} if it has no infinite accepting runs.

\subsection{Composition}
\label{subsec_composition}

We define an asynchronous (interleaving-based)
parallel composition operator with rendezvous
synchronization.
Given two automata
$A_1 = (Q_1,q_{0,1},I_1,O_1,T_1)$
and
$A_2 = (Q_2,q_{0,2},I_2,O_2,T_2)$,
the composition of $A_1$ and $A_2$, denoted $A_1 \| A_2$,
is defined, provided $O_1\cap O_2=\emptyset$, as the automaton
$$
A_1 \| A_2 \defeq
	(Q_1\times Q_2, (q_{0,1},q_{0,2}),
		(I_1\cup I_2)\setminus (O_1\cup O_2), O_1\cup O_2, T)
$$
where $((q_1,q_2),x,(q_1',q_2'))\in T$ iff one of the following holds:
\begin{itemize}
\item $x\in O_1$ and $q_1\tr{x!}q_1'$ and either $x\in I_2$ and
  $q_2\tr{x?}q_2'$ or $x\not\in I_2$ and $q_2'=q_2$.
\item $x\in O_2$ and $q_2\tr{x!}q_2'$ and either $x\in I_1$ and
  $q_1\tr{x?}q_1'$ or $x\not\in I_1$ and $q_1'=q_1$.
\item $x\in (I_1\cup I_2)\setminus (O_1\cup O_2)$ and at least one of
  the following holds: (1) $x\in I_1\setminus I_2$ and
  $q_1\tr{x?}q_1'$ and $q_2'=q_2$, (2) $x\in I_2\setminus I_1$ and
  $q_2\tr{x?}q_2'$ and $q_1'=q_1$, (3) $x\in I_1\cap I_2$ and
  $q_1\tr{x?}q_1'$ and $q_2\tr{x?}q_2'$.
\end{itemize}

During composition, the product automaton $A_1 \| A_2$ ``inherits'' the
safety and liveness properties of each of its components.
Specifically, a product state $(q_1,q_2)$ is an error state if either $q_1$
or $q_2$ are error states.
A product state $(q_1,q_2)$ is an accepting state if either $q_1$
or $q_2$ is an accepting state.


Note that $\|$ is commutative and associative. So we can write $A_1
\| A_2 \|\cdots \| A_n$ without parentheses, for a set of $n$
automata.

We call a product $A_1 \| \cdots \| A_n$ \emph{strongly non-blocking}
if in each reachable global state $(s_1, ..., s_n)$, if some automaton
$A_i$ is willing to send a message $x$, and all automata $A_j$ which
accept $x$ are in non-output states, then these automata should be
able to synchronize on the transition $x$ in that global state. On the
other hand, we call a product \emph{weakly} non-blocking if in each
reachable global state $(s_1, ..., s_n)$, if some automaton $A_i$ is
willing to send a message $x$, then it is possible for all automata
which accept $x$ to eventually synchronize on the transition $x$. Note
that strong non-blockingness is a safety property and can be useful
when the model-checker cannot verify liveness properties.


\subsection{From Scenarios to Incomplete Automata}
\label{subsec_scenarios_to_automata}

The first step in our synthesis method is to automatically generate
from the set of input scenarios an {\em incomplete automaton} for
each protocol process. The second step is then to complete these
incomplete automata to derive a complete protocol. In the sections
that follow, we formalize and study the automata completion problem.
In this section, we illustrate the first step of going from scenarios
to incomplete automata, by means of the ABP example.

The idea for transforming scenarios into incomplete automata is
simple.  First, for every ``swim lane'' in the message sequence chart
corresponding to a given scenario, we identify the corresponding
automaton in the overall system. For example, in each scenario shown
in Figure~\ref{fig_abp_scenarios}, the left-most lane corresponds to
ABP Sender and the right-most lane to ABP Receiver.  These scenarios
omit the environment processes for simplicity.  In particular channel
processes are omitted, however, we will use a primed version of a
message when referencing it on the process that receives it.

Second, for every protocol process $P$, we generate an incomplete automaton
$A_P$ as follows.
For every message {\em history} $\rho$
(i.e., finite sequence of messages received or sent by the process) 
specified in some scenario in the lane for $P$, we identify a state
$s_\rho$ in $A_P$.
If $\rho' = \rho \cdot x$ is an extension of history $\rho$ by one
message $x$, then there is a transition $s_\rho \tr{x} s_{\rho'}$ in $A_P$.
Applying this procedure to the four scenarios of Figure~\ref{fig_abp_scenarios},
we obtain the two incomplete automata shown in Figure~\ref{fig_abp_pta}.

\input{abp_pta}

Third, scenarios are annotated with labels. As shown in the first
scenario of Figure~\ref{fig_abp_scenarios}, labels appear between
messages on swim lanes. These are used to merge the states that
correspond to message histories that are followed by the same label.
Merging occurs for states of a single scenario as well as across
multiple ones if the same label is used in different scenarios.
If consistent labels are given to the initial and final positions
in all swim lanes of the scenarios the resulting incomplete automata
can be made cyclic. Furthermore, labels are essential for
specifying recurring behaviors in scenarios and the structure of the
incomplete automaton depends on the number and positions of labels used.

Finally, it is often the case that different behaviors of a system are
equivalent up to simple replacement of messages. For example, all the
ABP scenarios express valid behaviors if $p_0$ and $a_0$ messages are
consistently replaced with $p_1$ and $a_1$ messages respectively and
vice versa.  Thus, our framework allows for scenarios to be
characterized as ``symmetric''.

\begin{figure}[!t]
  \centering
  \resizebox{0.75\textwidth}{!}{
    \begin{tikzpicture}[node distance=1.5cm,->,>=stealth',shorten >=1pt, align=center,
      every node/.style={on grid}]
      \node[state,initial,initial text={}] (s0) {};
      \node[state] (s1) [right=1.6cm of s0] {};
      \node[state, inner sep=2pt] (s2) [right=of s1] {$q_1$};
      \node[state, inner sep=2pt] (s3) [right=1.8cm of s2] {$q_2$};
      \node[state, inner sep=2pt] (s4) [right=1.8cm of s3] {$q_5$};
      \node[state, inner sep=2pt] (s5) [below right=of s2] {$q_3$};
      \node[state, inner sep=2pt] (s6) [right=of s5] {$q_4$};
      \node[state, right=of s4] (s7) {};
      \node[draw=none, below right=1.4cm of s0, rotate=-45] {before sending 0};
      \node[draw=none, below right=1.5cm of s4, rotate=-45] {before sending 1};
      \path
      (s0) edge node[above] {$send!$} (s1)
      (s1) edge node[above] {$p_0!$} (s2)
      (s2) edge node[above] {$a_1'?$} (s3)
      (s3) edge node[above] {$a_0'?$} (s4)
      (s2) edge[out=90, in=90, looseness=.8] node[above] {$a_0'?$} (s4)
      (s2) edge node[below left] {$timeout?$} (s5)
      (s6) edge node[below right] {$a_0'?$} (s4)
      (s5) edge node[above] {$p_0!$} (s6)
      (s4) edge node[above] {$send!$} (s7)
      ;
    \end{tikzpicture}}
  \caption{Incomplete protocol automaton for ABP Sender after adding
    symmetric scenarios and merging labeled states. (Only the first
    half of the automaton is shown, the rest is the symmetric case for
    packet 1.)}
  \label{fig_final_incomplete_abp_sender}
\end{figure}

We annotate the swim lanes of the ABP Sender scenarios of
Figure~\ref{fig_abp_scenarios} with ``before sending 0'' and ``before
sending 1'' labels, and the swim lanes of the ABP Receiver with
``before receiving 0'' and ``before receiving 1'' labels.
We also add the symmetric scenarios by switching 0 messages with 1 messages.
The resulting incomplete automaton for ABP Sender is shown in
Figure~\ref{fig_final_incomplete_abp_sender}.


\subsection{Automata Completion}

Having transformed the input scenarios into incomplete automata, the next
step is to {\em complete} those automata by adding the appropriate transitions,
so as to synthesize a complete and correct protocol. In this section we
formalize this completion problem. We define two versions of the problem:
a general version (Problem~\ref{probgeneral}) and a special version
with only a single incomplete automaton (Problem~\ref{probsimple}).
We will use Problem~\ref{probsimple} in Section~\ref{sec_complexity} to show
that even in the simplest case automaton completion is combinatorially hard.

Consider an automaton $A = (Q,q_0,I,O,T)$. 
Given a set of transitions $T' \subseteq Q\times (I\cup O)\times Q$,
the {\em completion of $A$ with $T'$} is the new automaton
$A' = (Q,q_0,I,O,T\cup T')$. 


\begin{problem}
\label{probsimple}
Given automaton $E$ (the {\em environment}) and deterministic automaton
$P$ (the {\em process}) such that $E \| P$ is defined,
find a set of transitions $T$ such that,
if $P'$ is the completion of $P$ with $T$, then $P'$ is
deterministic and $E \| P'$ has no reachable deadlock states.
\end{problem}
Note that if $E \| P$ is defined then $E \| P'$ is also defined,
because, by definition, completion does not modify the interface (sets
of inputs and outputs) of an automaton.

\begin{problem}
\label{probgeneral}
Given a set of {\em environment automata} $E_1,...,E_m$ and set
of deterministic {\em process automata} $P_1,...,P_n$ such that
$E_1\|\cdots \| E_m \| P_1\|\cdots\| P_n$ is defined,
find sets of transitions $T_1,...,T_n$ such that,
if $P_i'$ is the completion of $P_i$ with $T_i$, for $i=1,...,n$, then
\begin{itemize}
\item $P_i'$ is deterministic, for $i=1,...,n$,
\item if the product automaton $\Pi := E_1\|\cdots \| E_m \|
  P_1'\|\cdots\| P_n'$ is a safety automaton then it is safe,
\item if $\Pi$ is a liveness automaton then it is empty,
\item $\Pi$ has no reachable 
	deadlock states, 
\item and, optionally, $\Pi$ is weakly (or strongly) non-blocking.
\end{itemize}
\end{problem}

Some of the environment processes $E_i$ can be be safety or liveness
monitors.  The last requirement means that Problem~\ref{probgeneral}
comes in three versions, one where strong non-blockingness is
required, one where the weak version is required, and a third where
none is. These are options provided by the user.

%% file: abp_pta.tex
\begin{figure}[!t]
  \begin{minipage}{\textwidth}
  \centering
  \resizebox{0.95\textwidth}{!}{
  \begin{tikzpicture}[node distance=1.5cm,->,>=stealth',shorten >=1pt,align=center]
    \node[state,initial by arrow,initial text=] (s0) {};
    \node[state] (s05) [right of=s0] {};
    \node[state] (s1) [right of=s05] {};
    \node[state] (s2) [right of=s1] {};
    \node[state] (s25) [right of=s2] {};
    \node[state] (s3) [right of=s25] {};
    \node[state] (s4) [right of=s3] {};
    \node[state] (s45) [right of=s4] {};
    \node[state] (s5) [right of=s45] {};
    \node[state] (s6) [right of=s5] {};
    \node[state] (s7) [below of=s3] {};
    \node[state] (s8) [right of=s7] {};
    \node[state] (s9) [right of=s8] {};
    \node[state] (s95) [right of=s9] {};
    \node[state] (s10) [right of=s95] {};
    \node[state] (s11) [right of=s10] {};
    \node[state] (s12) [below of=s10] {};
    \node[state] (s13) [right of=s12] {};
    \path (s0) edge node[above] {$send!$} (s05)
          (s05) edge node[above] {$p_0!$} (s1)
          (s1) edge node[above] {$a_0'?$} (s2)
          (s2) edge node[above] {$send!$} (s25)
          (s25) edge node[above] {$p_1!$}  (s3)
          (s3) edge node[above] {$a_1'?$}  (s4)
          (s4) edge node[above] {$send!$} (s45)
          (s45) edge node[above] {$p_0!$} (s5)
          (s5) edge node[above] {$a_0'?$} (s6)
          (s3) edge node[left]  {$timeout?$} (s7)
          (s7) edge node[above] {$p_1!$}  (s8)
          (s8) edge node[above] {$a_1'?$}  (s9)
          (s9) edge node[above] {$send!$} (s95)
          (s95) edge node[above] {$p_0!$} (s10)
          (s10) edge node[above] {$a_0'?$} (s11)
          (s10) edge node[left] {$a_1'?$} (s12)
          (s12) edge node[above] {$a_0'?$} (s13)
          ;
  \end{tikzpicture}}
  \end{minipage}

  \begin{minipage}{\textwidth}
  \centering
  \resizebox{1.0\textwidth}{!}{
  \begin{tikzpicture}[node distance=1.6cm,->,>=stealth',shorten >=1pt, align=center]
    \node[state,initial by arrow,initial text=] (r0) {};
    \node[state] (r1) [right of=r0] {};
    \node[state] (r15) [right of=r1] {};
    \node[state] (r2) [right of=r15] {};
    \node[state] (r3) [right of=r2] {};
    \node[state] (r35) [right of=r3] {};
    \node[state] (r4) [right of=r35] {};
    \node[state] (r5) [right of=r4] {};
    \node[state] (r55) [right of=r5] {};
    \node[state] (r6) [right of=r55] {};
    \node[state] (r7) [below of=r4] {};
    \node[state] (r8) [right of=r7] {};
    \node[state] (r9) [right of=r8] {};
    \node[state] (r95) [right of=r9] {};
    \node[state] (r10) [right of=r95] {};
    \path (r0) edge node[above] {$p_0'?$} (r1)
          (r1) edge node[above] {$deliver!$} (r15)
          (r15) edge node[above] {$a_0!$} (r2)
          (r2) edge node[above] {$p_1'?$}  (r3)
          (r3) edge node[above] {$deliver!$} (r35)
          (r35) edge node[above] {$a_1!$}  (r4)
          (r4) edge node[above] {$p_0'?$} (r5)
          (r5) edge node[above] {$deliver!$} (r55)
          (r55) edge node[above] {$a_0!$} (r6)
          (r4) edge node[left]  {$p_1'?$}  (r7)
          (r7) edge node[above] {$a_1!$}  (r8)
          (r8) edge node[above] {$p_0'?$} (r9)
          (r9) edge node[above] {$deliver!$} (r95)
          (r95) edge node[above] {$a_0!$} (r10)
    ;
  \end{tikzpicture}}
  \end{minipage}
  \caption{\label{fig_abp_pta} Incomplete protocol automata for ABP: Sender (top), Receiver (bottom).}
\end{figure}

%% file: solutions.tex
\section{Solving Automata Completion}

In this section, we consider procedures to solve the automata
completion problem.  First, we show that Problems~\ref{probsimple} and
\ref{probgeneral} are NP-complete and PSPACE-complete
respectively. Then, we present an explicit search algorithm that
eagerly prunes parts of the search space and a heuristic which ranks
candidate completions. Finally, we describe an algorithm which reduces
automata completion to a model-checking problem for a symbolic
model-checker.

\subsection{Complexity}
\label{sec_complexity}

It can be shown that Problem~\ref{probgeneral} is
PSPACE-complete. Note that this is not surprising, as the verification
problem itself is PSPACE-complete, for safety properties of
distributed protocols. However, in the special case of two processes,
while verification can be performed in polynomial time, a reduction
from 3-SAT shows that the corresponding completion
Problem~\ref{probsimple} is NP-complete. The proofs can be found in
the appendix.

\begin{theorem}\label{thm:complexity}
Problem~\ref{probgeneral} is PSPACE-complete and
Problem~\ref{probsimple} is NP-complete.
\end{theorem}

\subsection{Explicit search}
\label{sec_search}

This algorithm for solving the automata completion problem
(Problem~\ref{probgeneral}) is based on an explicit search over the
space of possible completions, guided by various heuristics.  More
specifically, the algorithm explores a {\em search tree} in which every
node is a set of added transitions $T$ (we include in $T$ the
transitions added in all incomplete protocol automata). The children
of each node $T$ are those nodes $T^{\prime} \supset T$ which contain
exactly one more transition than $T$.  The root of the tree is the
empty set of transitions, which corresponds to the original input
(i.e., the incomplete automata generated from the scenarios).

For every newly visited node $T$ (including the root) a model-checking problem
is solved: we form the product of all environment processes, protocol
processes (with the added transitions $T$), and monitors, and we check 
the absence of deadlocks, safety, and liveness violations 
(and optionally also non-blockingness).
The following cases are possible:
\begin{enumerate}
\item No violations are found. In this case, $T$ is a correct solution,
	and the search terminates.
\item A safety or liveness violation is found. This means that this
	candidate solution $T$ is incorrect. Moreover, any other candidate
	$T'$ obtained by adding extra transitions to $T$, i.e., $T'\supseteq T$,
	will also be incorrect, by exhibiting the same violations. 
	This is because adding extra local transitions can only add, 
	but not remove, global transitions. This in turn implies 
	that any reachable error state with $T$ will also be a reachable
	error state with $T'$, so any safety violation with $T$ will also
	be a safety violation with $T'$. Similarly, any reachable accepting
	cycle with $T$ will also be a reachable accepting cycle with $T'$,
	so liveness violations cannot be removed either.
	In conclusion, in this case, the entire sub-tree under $T$ can be
	pruned from the search.
\item No safety nor liveness violation is found, but a deadlock or
\label{case_continue}
	blocking state is found. In this case, $T$ is incorrect, but could
	potentially be made correct by adding more transitions. The search
	continues exploring the children of $T$.
\end{enumerate}
The search algorithm saves every visited node $T$. The same node might be
visited via different paths. For example, adding first $t_1$, then $t_2$,
leads to the same node as adding first $t_2$, then $t_1$. 
To reduce the search space, the search stops (and backtracks) when it finds
a node that has already been visited.
This is clearly sound and complete.

The search continues only in Case~\ref{case_continue}.  In this case,
we use heuristics to determine in which order should the children of
$T$ be explored. The heuristic we used prioritizes the children of $T$
according to how similar they are to existing transitions in the
protocol automaton. We deem two transitions as being similar if their
message and destination are the same and if their starting states
already agree on some other transition.  In other words, if two states
handle a message by transitioning to the same state, or
indistinguishable states, the heuristic extrapolates that the states
should also handle other messages in the same manner.  For example, in
Figure~\ref{fig_final_incomplete_abp_sender}, states $q_1$ and $q_4$
both handle message $a_0'$ by transitioning to state $q_5$.  Hence,
the heuristic prioritizes the candidate transition from $q_4$ to $q_3$
on $timeout$, over say $q_4$ to $q_2$, since $q_1$ also handles
$timeout$ by transitioning to $q_3$.  Note that this transition
correctly generalizes the behavior on a single timeout described in
the scenarios to multiple timeouts.


\subsection{BDD-based Symbolic Computation}
This technique reduces the automata completion problem to an instance
of a model-checking query, which is then solved by using BDD-based
symbolic model-checking techniques. Consider a set of environment
automata $\{E_1, E_2, \ldots, E_m\}$ and a set of (possibly
incomplete) deterministic process automata $\{P_1, P_2, \ldots
P_n\}$, with each $P_i = (Q_i, q_{0,i}, I_i, O_i, T_i)$, as described
in Section~\ref{sec_formalization}. For each state $q_i \in Q_i$ and
for each event $x \in I_i \cup O_i$ such that $q_i \tr{x} q_i' \notin
T$ for any $q_i' \in Q_i$, we introduce a variable $t_{q_i, e}$ whose value
ranges over the set $Q_i \cup \{ \mathsf{\bot} \}$. Intuitively, these
variables encode all possible ways to complete the transition relation
$T_i$, including the possibility that no transition exists. Each of
the transition relations $T_i$ is now parametrized by
the newly introduced variables $t_{q_i, e}$. Let $P_i'$ be the
automata obtained by replacing the transition relation of $P_i$ with
the parametrized transition relation whose construction we have just
described.

We denote the composition $E_1 \| E_2 \| \ldots \| E_m \| P_1' \| P_2'
\| \ldots \| P_n'$ by $\mathcal{P}$ and its transition relation by
$\mathcal{T}$. Note that $\mathcal{T}$ is also parametrized by the
newly introduced $t_{q_i, e}$ variables. Also, the original
composition $E_1 \| E_2 \| \ldots \| E_m \| P_1 \| P_2 \ldots \| P_n$
has only one initial state, by definition, and thus, $\mathcal{P}$ has
only one (parametrized) initial state as well.

Suppose we are given a safety monitor with a set of error states. We
can symbolically represent the states where the monitor automaton is
in an error state by a propositional formula $\varphi$. Now, the
parameter values such that states satisfying $\neg \varphi$ are not
reachable in the composition $\mathcal{P}$ can be obtained by
model-checking $\mathcal{P}$ with the CTL property, $\tle\tlf
\neg\varphi$. If this property is found to be
true, then an erroneous state in the monitor is reachable for every
valuation of the parameters. If the property
is found to be false, then there must exist a valuation for the
parameters which prevents a state satisfying $\neg\varphi$ from being
reachable. These parameter values represent a completion that
satisfies the property that no state satisfying $\neg\varphi$ is
reachable.

\subsubsection{Determinism, Deadlock Freedom, Non-blockingness and
  Liveness.}
We encode constraints for determinism by restricting the set of
initial values that the parameters can take. 

A deadlock state is characterized by the formula $\mathsf{DL} =
\bigwedge_{o \in \bigcup_i O_i} (\neg \mathsf{sync}(o))$, where
$\mathsf{sync}(o)$ is a formula which expresses that the (sole) sender
of the event $o$ is in a state where it can send $o$, and all the
receivers of event $o$ are in states where they can receive an
$o$. Thus the formula $\tle\tlf(\mathsf{DL})$ represents the set of
states which eventually deadlock. A valid completion would render
these states unreachable.

Suppose we are given a liveness automaton with a set $Q_a$ of
accepting states. A completion which is live needs to have no runs
that visit states in $Q_a$ infinitely often. If all states in $Q_a$
are represented symbolically by the propositional formula $Q_a^s$,
then the set of states from which it is possible to visit states in
$Q_a$ infinitely often can be characterized by the CTL formula $\psi =
\tla\tlg \tle\tlf Q_a^s$. Again, we desire that these states be
unreachable in a valid completion.

The non-blockingness requirement can expressed as a safety requirement
$\mathsf{NonBlock}$, and hence can be handled in a similar
manner. Finally, we use a symbolic model-checker (NuSMV~\cite{NuSMV2})
to check if $\mathcal{P}$ satisfies
$\tle\tlf(\neg\varphi \vee \mathsf{DL} \vee \neg\mathsf{NonBlock}
\vee\psi)$. The valuation of parameters for which $\mathcal{P}$ does
not not satisfy this property represents the completion which
has the required safety, liveness, and deadlock-freedom
properties.

%% file: case_studies.tex
\section{Evaluation}
\label{sec_evaluation}
We investigate (1) how effective scenarios are in
reducing the empirical complexity of the automata completion
problem, (2) the amount of generalization that the proposed algorithms
are able to perform, and (3) how adding scenarios reduces the
number of formal specifications required for successful completion.

\subsubsection{Synthesis with no scenarios.}
To validate our hypothesis that scenarios make the synthesis problem
easier, we attempted to synthesize the ABP protocol with no
transitions specified, but with bounds on the number of states of the
processes. These bounds were set to be equal to the corresponding
number of states in the manually constructed version of the ABP
protocol. We required that the protocol satisfy all the properties
discussed in Section~\ref{sec_methodology}.

The BDD-based symbolic algorithm ran out of memory\footnote{We used a
memory limit of 16GB and a time limit of one hour.} and failed to
synthesize a protocol. Recall that the heuristic algorithm performs
generalization by using a similarity metric between states. When the
starting point is an empty transition relation, there is no similarity
between states, and the heuristic fails to differentiate between
candidate transitions. The resulting search procedure runs out of
time.

\subsubsection{Varying the number of scenarios.}
When all four of the scenarios in Figure~\ref{fig_abp_scenarios} are
used, both the explicit search algorithm and the BDD-based search are
able to find a correct completion for the protocol. Moreover, both the
algorithms find a correct completion, when applied to just \emph{one}
scenario. A quantitative summary of our experiments can be found in
Table~\ref{tab:experimental_results}.  We applied our algorithms on
the incomplete automata constructed from the first scenario (row 1),
the second scenario (row 2), and all four scenarios (row 3).  For each
case, we report the number of states of the incomplete automata, the
number of transitions in the completions found by the algorithms, and
their computational requirements.

In the case of all four scenarios, the ranking heuristic described in
Section~\ref{sec_search} chooses a candidate transition that is part
of the correct completion at every step.  The search does not
backtrack, does not prune any nodes, and takes less than a minute to
complete.  When only the first scenario is used, the ranking heuristic
is less effective in the same way as when no scenarios were used.
However, the additional structure imposed by the scenarios
significantly constrains the search space and the explicit search with
pruning successfully returns a completion.  The incomplete automata
that correspond to scenario 2 include intermediate states that the
heuristic can use to successfully rank candidate edges. The runtime of
the explicit search algorithm varies with the order in which candidate
transitions are explored. We report the 75th percentile over several
randomly chosen orders.

In contrast, the BDD-based search performs better when a single
scenario is used rather than all four.  The intuition behind this is
that the number of BDD variables is smaller when there are fewer
intermediate states in the incomplete automata.  As a result, the
search finishes in less than 15 seconds when only the first scenario
of Figure~\ref{fig_abp_scenarios} is used, less than 15 minutes when
only the second scenario is used, and less than 35 minutes when all
four are used.  Synthesis takes longer with the second scenario
because the incomplete sender automaton, as mentioned earlier, has more
intermediate states than the first scenario.

\begin{table}[!t]
  \caption{Quantitative summary of experiments.}

  \setlength{\extrarowheight}{2pt}
  \begin{center}
  \begin{tabular}{|>{\centering}m{0.15\textwidth}|>{\centering}p{0.11\textwidth}|>{\centering}p{0.11\textwidth}|>{\centering}p{0.15\textwidth}|>{\centering}m{0.12\textwidth}|>{\centering}m{0.12\textwidth}|>{\centering}m{0.12\textwidth}|}
  \hline
  \multirow{2}{0.15\textwidth}{} &
  \multicolumn{2}{>{\centering}m{0.22\textwidth}|}{Number of states in
    incomplete automaton} &
  \multirow{2}{0.15\textwidth}{\centering{}Number of transitions to be
    completed} & \multicolumn{3}{c|}{Computational
    requirements}\tabularnewline
  \cline{2-3} \cline{5-7} 
   & Sender & Receiver &  & Heuristic time & BDD time & BDD memory \tabularnewline
  \hline
  \hline
  Scenario 1 & 6 & 6 & 6 & 4 min. & 15 sec. & 100MB \tabularnewline
  \hline
  Scenario 2 & 10 & 6 & 8 & 30 sec. & 15 min. & 1GB \tabularnewline
  \hline
  All scenarios & 12 & 8 & 8 & 45 sec. & 35 min. & 3.8GB \tabularnewline
  \hline
  \end{tabular}
  \end{center}
  \label{tab:experimental_results}
\end{table}





\subsubsection{Generalization and inference of unspecified behaviors.}


With just one scenario specified, the algorithms successfully perform
the generalization required to obtain a correct completion. We believe
that the generalization performed is non-obvious: the correct protocol
behaviors on packet loss, loss of acknowledgments and message
duplication are inferred, even though the scenario does not describe
what needs to happen in these situations. As can be seen in
Figure~\ref{fig_first_scenario_incomplete_automata}, the incomplete
automata constructed from the scenario only describe the protocol
behavior over lossless channels. The algorithms are guided solely by
the liveness and safety specifications to infer the correct behavior.
In contrast, when all four scenarios are specified, the scenarios
already contain information about the behavior of the protocol when a
single packet loss or a single message duplication occurs. The
algorithm thus needs to only generalize this behavior to handle an
arbitrary number of packet losses and message duplications.


\subsubsection{Varying the Correctness Requirements.}
We observed that when fewer scenarios were used, we needed to specify
more properties --- some of which were non-obvious --- so that the
algorithms could converge to a correct completion. For instance, when
only one scenario was specified, we needed to include the liveness
property that every deliver message was eventually followed by a send
message. Owing to the structure of the incomplete automata, this
property was not necessary to obtain a correct completion when all
four scenarios were specified. Another property which was necessary to
reject trivial completions when no scenarios were specified was that
there has to be at least one send message in every run. Therefore, in
some cases, using scenarios can compensate for the lack of detailed
formal specifications.

\subsubsection{Discussion of Experimental Results.}
The experimental results clearly demonstrate that specifying the
behavior of the protocol using scenarios is essential for the
algorithms we have evaluated to be able to construct a correct
completion. In particular, even providing just one scenario allows the
algorithms to converge on a correct completion within a reasonable
amount of time. In contrast, none of the approaches we have evaluated
were successful in synthesizing the required protocol when no
scenarios were provided. An interesting trend that we observed was
that the heuristic algorithm and the BDD-based symbolic algorithm seem
to complement each other. Specifically, the BDD-based symbolic
algorithm was effective when the number of scenarios --- and therefore
the number of states in the incomplete automata --- was small. On the
other hand, the heuristic search does better when more information is
provided in the form of additional scenarios --- which in turn causes
a larger number of states in the incomplete automata to be
similar. This is because the BDD-based symbolic algorithm essentially
constructs BDDs for \emph{all} possible completions and picks one that
satisfies the properties. In general, the number of possible
completions increases with the number of states in the automata, which
explains why the BDD-based algorithm performs well when the number of
states in the automaton is small. The heuristic algorithm exploits
similarities between states to converge on a correct completion. When
the algorithm is provided with the incomplete automaton from the first
scenario --- which has a minimal number of states --- it is unable to
find states which are similar and thus degenerates to an exhaustive
search over all completions.  Finally, we also observe that additional
scenarios can compensate for unspecified correctness properties. This
frees the protocol designer from having to specify the requirements
completely formally.

\begin{figure}[!t]
  \centering
  \begin{minipage}{.4\textwidth}
    \centering
    \resizebox{\textwidth}{!}{
    \begin{tikzpicture}[node distance=2cm,->,>=stealth',shorten >=1pt, align=center]
      \node[state,initial,initial text={}] (s0) {};
      \node[state] (s1) [above right=.3cm and 1cm of s0] {};
      \node[state] (s2) [right=of s1] {};
      \node[state] (s3) [below right=.3cm and 1cm of s2] {};
      \node[state] (s4) [below left=.3cm and 1cm of s3] {};
      \node[state] (s5) [left=of s4] {};
      \path
      (s0) edge node[above left] {$send!$} (s1)
      (s1) edge node[above] {$p_0!$} (s2)
      (s2) edge node[above right] {$a_0'?$} (s3)
      (s3) edge node[below right] {$send!$} (s4)
      (s4) edge node[above] {$p_1!$} (s5)
      (s5) edge node[below left] {$a_1'?$} (s0)
      ;
    \end{tikzpicture}}
  \end{minipage}\qquad
  \begin{minipage}{.35\textwidth}
    \centering
    \resizebox{\textwidth}{!}{
    \begin{tikzpicture}[node distance=3cm,->,>=stealth',shorten >=1pt,align=center]
      \node[state,initial,initial text={}] (r0) {};
      \node[state, above right=.3cm and 1cm of r0] (r1) {};
      \node[state, right=1.4cm of r1] (r2) {};
      \node[state, below right=.3cm and 1cm of r2] (r3) {};
      \node[state, below left=.3cm and 1cm of r3] (r4) {};
      \node[state, left=1.4cm of r4] (r5) {};
      \path
      (r0) edge node[above left] {$p_0'?$} (r1)
      (r1) edge node[above] {$deliver!$} (r2)
      (r2) edge node[above right] {$a_0!$} (r3)
      (r3) edge node[below right] {$p_1'?$} (r4)
      (r4) edge node[above] {$deliver!$} (r5)
      (r5) edge node[below left] {$a_1!$} (r0)
      ;
    \end{tikzpicture}}
  \end{minipage}
  \caption{\label{fig_first_scenario_incomplete_automata} Incomplete automata constructed from the first scenario of Figure~\ref{fig_abp_scenarios}.}
\end{figure}
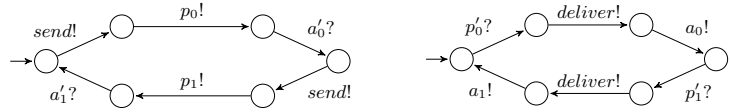

%% file: conclusions.tex
\section{Conclusions}
The main contribution of this paper is a new methodology, supported by
automatic synthesis techniques, for specifying finite-state
distributed protocols using a mix of representative behaviors and
correctness requirements. The synthesizer derives a skeleton of the
state machine for each process using the states that appear in the
scenarios and then finds a completion that satisfies the
requirements. The promise of the proposed method is demonstrated by
the ability of the synthesizer to learn the correct ABP protocol from
just a single scenario corresponding to the typical case. Future
research should focus on the scalability of the algorithms for
protocol completion. One idea is to heuristically limit the choices
for targets of transitions before applying the BDD-based symbolic
computation, and another approach is to reduce the number of states by
representing the implementation as an \emph{extended} FSM with
variables.

%% file: appendix.tex
\appendix
\section*{Appendix: Complexity of Automata Completion}

\begin{theorem}\label{claim:probsimple_np}
  Problem~\ref{probsimple} is in NP.
\end{theorem}
\begin{proof}
  Problem~\ref{probsimple} is in NP since we can guess a
completion, and then check whether the requirements of the problem
are satisfied. Checking whether the
resulting process automaton is deterministic can be done in polynomial time.
Checking whether the product automaton is deadlock-free
can be done in polynomial time also, because there are only two automata in
the product.
\end{proof}

\begin{theorem}
  Problem~\ref{probsimple} is NP-complete.
\end{theorem}

\begin{proof}
We will show that 3SAT is reducible to Problem~\ref{probsimple}.
Together with Theorem~\ref{claim:probsimple_np} this shows that Problem~\ref{probsimple} is NP-complete.

Let $U = \{u_1,u_2,...,u_n\}$ be a set of variables and $C = \{c_1,c_2,...,c_m\}$
be a set of clauses, such that $|c_i|=3$ for $1\le i\le m$, making up an arbitrary instance of 3SAT.

We write $z^1_j$, $z^2_j$, $z^3_j$ for the literals of $c_j$.
We write $v^i_j$ for the variable of literal $z^i_j$.
Note that $z^i_j$ can be equal to $v^i_j$ or its negation $\overline{v^i_j}$.

We construct a process automaton $P$ and an environment $E$ with $E\| P$ defined, such that the following are equivalent:
\begin{itemize}
\item there is a set of transitions $T$ such that if $P'$ is the completion of $P$ with $T$, then $P'$ is deterministic and $E\| P'$ has no reachable deadlock states 
\item there is a truth assignment for $U$ that satisfies all clauses in $C$.
\end{itemize}

The process automaton has an initial state $q_0$, and a pair of states $q^t_i$ and $q^f_i$
for each variable $u_i$. At each step, the environment challenges the process to
instantiate some variable $u_i$ by transmitting the message $x^D_i$. On receipt of
this message, the completed process has to transition to one of $q^t_i$ and $q^f_i$.
It responds with either $x^t_i$ or $x^f_i$ indicating the assignment made to $u_i$.
The environment performs this challenge for each literal in each clause, and enters
a deadlock state if any clause is left unsatisfied. It follows that a completion
exists iff $C$ is satisfiable.

P is an automaton $(Q^P,q^P_0,I^P,O^P,T^P)$ such that:
\begin{itemize}
\item $Q^P = \{q^P_0\} \cup \{q^{t}_i, q^f_i \st i \in [1,n]\}$
\item $I^P = \{x_s\} \cup \{x^D_i \st i\in [1,n]\}$
\item $O^P = \{x^t_i, x^f_i \st i\in [1,n]\}$
\item $T^P = \{(q^t_i, x^t_i, q^P_0), (q^f_i, x^f_i, q^P_0) \st i\in [1,n]\} \cup \{(q^P_0, x_s, q^P_0)\}$
\end{itemize}

E is an automaton $(Q^E,q^E_0,I^E,O^E,T^E)$ such that:
\begin{itemize}
\item $Q^E = \{deadlock, success\} \cup \{q^D_{i,j}, q^V_{i,j}) \st \text{$i\in \{1,2,3\}$ and $j\in [1,m]$}\}$
\item $q^E_0 = q^D_{1,1}$
\item $I^E = \{x^t_i, x^f_i \st i\in [1,n]\}$
\item $O^E = \{x_s\} \cup \{x^D_i \st i\in [1,n]\}$
\item
 $T^E$ is the smallest set such that the following hold:
\begin{itemize}
\item
for $i\in \{1,2\}$ and $j\in[1,m-1]$, $(q^V_{i,j}, x^f_k, q^D_{i+1,j}), (q^V_{i,j}, x^t_k, q^V_{1,j+1}) \in T_E$ if $z^i_j = u_k$,
\item
for $i\in \{1,2\}$, $(q^V_{i,m}, x^f_k, q^D_{i+1,m}), (q^V_{i,j}, x^t_k, success) \in T_E$ if $z^i_m = u_k$,
\item
$(q^V_{3,m}, x^t_k, success)$ and $(q^V_{3,m}, x^f_k, deadlock)$ if $z^i_m = u_k$,
\item
for $i\in \{1,2\}$ and $j\in[1,m-1]$, $(q^V_{i,j}, x^t_k, q^D_{i+1,j}), (q^V_{i,j}, x^f_k, q^V_{1,j+1}) \in T_E$ if $z^i_j = \overline{u_k}$,
\item
for $i\in \{1,2\}$, $(q^V_{i,m}, x^t_k, q^D_{i+1,m}), (q^V_{i,j}, x^f_k, success) \in T_E$ if $z^i_m = \overline{u_k}$,
\item
$(q^V_{3,m}, x^f_k, success)$ and $(q^V_{3,m}, x^t_k, deadlock)$ if $z^i_m = \overline{u_k}$, and
\item
$(success, x_s, success) \in T_E$.
\end{itemize}
\end{itemize}

Note that $|Q^P| = |T^P| = 2\cdot n + 1$,
$|Q^E| = 6\cdot m + 2$, and $|T^E| = 9 \cdot m + 1$,
and $P$ and $E$ can be constructed in polynomial time.

For $U = \{u_1, u_2, u_3\}$, $C = \{c_1, c_2\}$, $c_1 = \{u_1, \overline{u_2}, u_3\}$,
and $c_2 = \{\overline{u_1}, \overline{u_2}, u_3\}$,
the process automaton $P$ and the environment automaton are shown in Figure~\ref{fig:process_automaton} and in Figure~\ref{fig:environment_automaton} respectively.

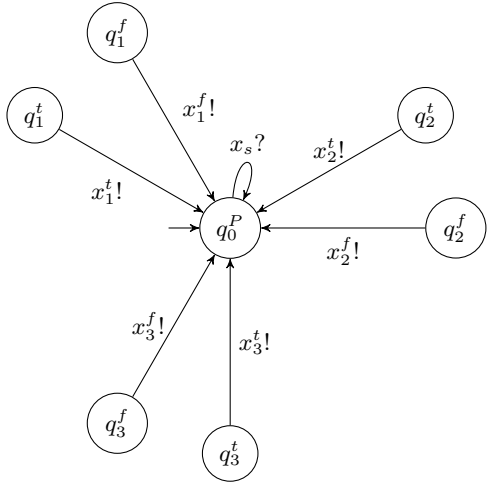
\begin{figure}[t]
  \centering
  \begin{tikzpicture}
    \path
    (0,0) node[initial, initial text={}, state] (initial) {$q^P_0$}
    +(150:3) node[state] (u1 decision true) {$q^t_1$}
    +(120:3) node[state] (u1 decision false) {$q^f_1$}
    (u1 decision true)
    edge[->] node[below left] {$x^t_1!$}
    (initial)
    (u1 decision false)
    edge[->] node[above right] {$x^f_1!$}
    (initial)
    ;
    \path
    (0,0)
    +(30:3) node[state] (u2 decision true) {$q^t_2$}
    +(0:3) node[state] (u2 decision false) {$q^f_2$}
    (u2 decision true)
    edge[->] node[above] {$x^t_2!$}
    (initial)
    (u2 decision false)
    edge[->] node[below] {$x^f_2!$}
    (initial)
    ;
    \path
    (0,0)
    +(-90:3) node[state] (u3 decision true) {$q^t_3$}
    +(-120:3) node[state] (u3 decision false) {$q^f_3$}
    (u3 decision true)
    edge[->] node[right] {$x^t_3!$}
    (initial)
    (u3 decision false)
    edge[->] node[left] {$x^f_3!$}
    (initial)
    (initial) edge[->,out=85,in=65,looseness=12] node[above] {$x_s?$} (initial)
    ;

  \end{tikzpicture}
  \caption{Process automaton before completion. Choosing an assignment to a boolean variable $u_i$ of the SAT problem means choosing whether to add an input transition labeled $x_i^D?$ from $q_0^P$ to $q_i^t$ (in which case $u_i$ is assigned to $\true$) or to $q_i^f$ (in which case $u_i$ is assigned to $\false$).}
  \label{fig:process_automaton}
\end{figure}

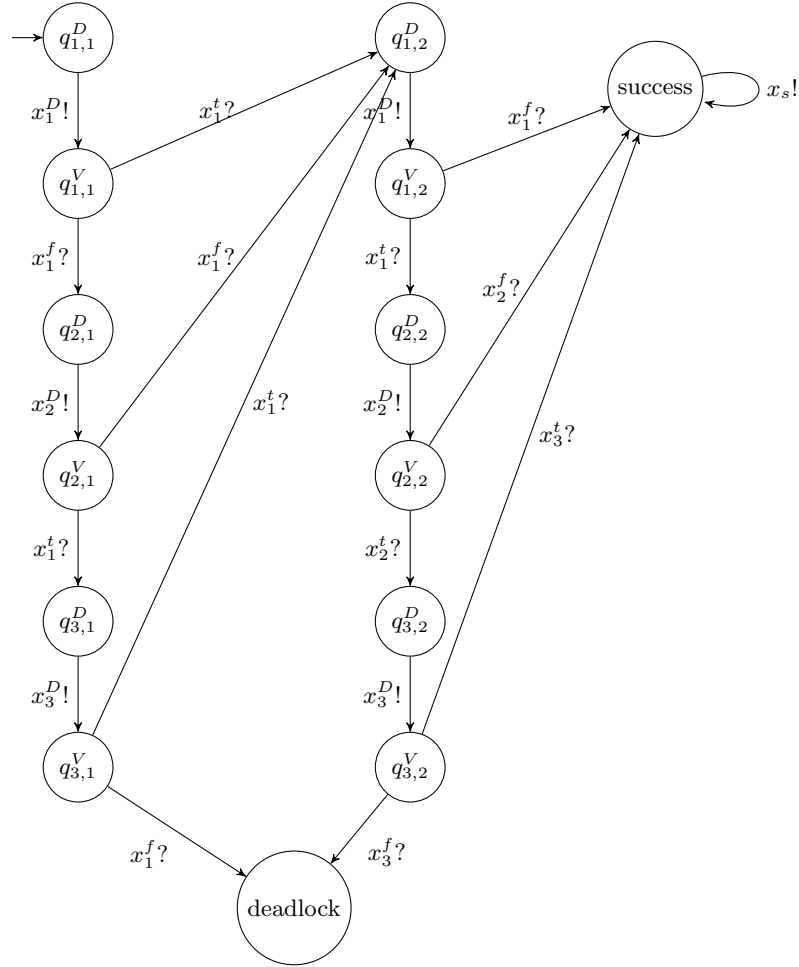
\begin{figure}[!t]
  \centering
  \resizebox{0.95\textwidth}{!}{
  \begin{tikzpicture}
    \path
    node (clause one) {Clause $c_1 = \{u_1, \overline{u_2}, u_3\}$}
    node[right=.8cm of clause one] (clause 2) {Clause $c_2 = \{\overline{u_1}, \overline{u_2}, u_3\}$}
    node[initial, initial text={}, state, below=.3cm of clause one] (ask u1) {$q^D_{1,1}$}
    node[state, below=of ask u1] (asked u1) {$q^V_{1,1}$}
    node[state, below=of asked u1] (ask u2) {$q^D_{2,1}$}
    node[state, below=of ask u2] (asked u2) {$q^V_{2,1}$}
    node[state, below=of asked u2] (ask u3) {$q^D_{3,1}$}
    node[state, below=of ask u3] (asked u3) {$q^V_{3,1}$}
    node[state, below right=1cm and 2cm of asked u3] (deadlock) {deadlock}

    node[state, below=.3cm of clause 2] (ask u1 c2) {$q^D_{1,2}$}
    node[state, below=of ask u1 c2] (asked u1 c2) {$q^V_{1,2}$}
    node[state, below=of asked u1 c2] (ask u2 c2) {$q^D_{2,2}$}
    node[state, below=of ask u2 c2] (asked u2 c2) {$q^V_{2,2}$}
    node[state, below=of asked u2 c2] (ask u3 c2) {$q^D_{3,2}$}
    node[state, below=of ask u3 c2] (asked u3 c2) {$q^V_{3,2}$}
    node[state, below right=of clause 2] (success) {success}

    (ask u1) edge[->] node[left] {$x^D_1!$} (asked u1)
    (asked u1) edge[->] node[left] {$x^f_1?$} (ask u2)
    (asked u1) edge[->] node[left] {$x^t_1?$} (ask u1 c2)

    (ask u2) edge[->] node[left] {$x^D_2!$} (asked u2)
    (asked u2) edge[->] node[left] {$x^t_1?$} (ask u3)
    (asked u2) edge[->] node[left] {$x^f_1?$} (ask u1 c2)

    (ask u3) edge[->] node[left] {$x^D_3!$} (asked u3)
    (asked u3) edge[->] node[below left] {$x^f_1?$} (deadlock)
    (asked u3) edge[->] node[right] {$x^t_1?$} (ask u1 c2)

    (ask u1 c2) edge[->] node[left] {$x^D_1!$} (asked u1 c2)
    (asked u1 c2) edge[->] node[left] {$x^t_1?$} (ask u2 c2)
    (asked u1 c2) edge[->] node[above] {$x^f_1?$} (success)

    (ask u2 c2) edge[->] node[left] {$x^D_2!$} (asked u2 c2)
    (asked u2 c2) edge[->] node[left] {$x^t_2?$} (ask u3 c2)
    (asked u2 c2) edge[->] node[left] {$x^f_2?$} (success)

    (ask u3 c2) edge[->] node[left] {$x^D_3!$} (asked u3 c2)
    (asked u3 c2) edge[->] node[below right] {$x^f_3?$} (deadlock)
    (asked u3 c2) edge[->] node[right] {$x^t_3?$} (success)

    (success) edge[loop right] node[right] {$x_s!$} (success)
    ;
  \end{tikzpicture}}
  \caption{Environment Automaton}
  \label{fig:environment_automaton}
\end{figure}

We now prove the correctness of the reduction:

\subsubsection{3SAT $\implies$ Problem~\ref{probsimple}: }

Assume that there is a truth assignment $t$ for $U$
that satisfies all clauses in $C$.

Let $T$ be the set of transitions such that
$(q^P_0, x^D_i, q^t_i) \in T$ iff $t(u_i) = \true$ and
$(q^P_0, x^D_i, q^f_i) \in T$ iff $t(u_i) = \false$.

Let $P'$ be the completion of $P$ with $T$.
It is easy to see that every transition leaving state $q^P_0$ has a distinct label,
since we add exactly one transition for each label $x^D_i$.
Therefore $P'$ is deterministic.
We now show that $E\| P'$ has no reachable deadlock states.

Let $R$ be the sequence of states in a run of $E\| P'$.
If $k$ is even then $R(k)$ is of the form $(q^D_{i,j}, q^P_0)$, $(success, q^P_0)$, or $(deadlock, q^P_0)$.
In the first case, $R(k)$ has exactly one outgoing edge and $R(k+1)$ is of the form $(q^V_{i,j}, q)$, where if $u_k = v^i_j$, $q$ is either $q^t_k$ or $q^f_k$. In the second case $R(k+1) = R(k)$ and in the third case $R(k)$ has no outgoing transitions.
By construction of $E$ and $P'$, $(q^V_{i,j}, q)$, where $q$ is either $q^t_k$ or $q^f_k$, has exactly one outgoing transition.
Therefore, if $(deadlock, q^P_0)$ is not reachable, then $E\| P'$ has no reachable deadlock states.

We assume that $(deadlock, q^P_0)$ is reachable and arrive at a contradiction.
If it is then there is $j \in [1,m]$ such that the following hold:
$(q^V_{1,j}, q_1) \tr{x_1!} (q^D_{2,j}, q^P_0)$,
$(q^V_{2,j}, q_2) \tr{x_2!} (q^D_{3,j}, q^P_0)$, and
$(q^V_{3,j}, q_3) \tr{x_3!} (deadlock, q^P_0)$,
where one of the following holds for each pair $q_i$ and $x_i$:
\begin{itemize}
\item $z^i_j = u_{k}$ and $q_i = q^f_k$ and $x_i = x^f_k$,
\item $z^i_j = \overline{u_k}$ and $q_i = q^t_k$ and $x_i = x^t_k$.
\end{itemize}

Furthermore, by the construction of $T$, $(q^V_{i,j}, q^f_k)$ is
reachable if $(q^P_0, x^D_k, q^f_k) \in T$ or if $t(u_k) = \false$,
and $(q^V_{i,j}, q^t_k)$ is reachable if $(q^P_0, x^D_k, q^t_k) \in T$
or if $t(u_k) = \true$.

We have shown that there is $j \in [1,m]$ such that
for every $i\in [1,3]$,
if $z^i_j = u_k$ then $t(u_k) = \false$ and
if $z^i_j = \overline{u_k}$ then $t(u_k) = \true$
which implies that $t(c_j) = \false$
and contradicts $t$ being a satisfying assignment.

\subsubsection{Problem~\ref{probsimple} $\implies$ 3SAT: }

Assume $T$ is a set of transitions, such that $P'$ is the completion of $P$ with $T$,
$P'$ is deterministic, and $E\| P'$ has no reachable deadlock states.

We construct a truth assignment $t$ for $U$ that satisfies all the clauses in $C$.

Since $E\| P'$ is deadlock free,
for every $i \in [1,3]$ and $j \in [1,m]$,
if a state $(q^V_{i,j}, q)$ is reachable,
because the only outgoing transitions from $q^V_{i,j}$ in $E$
are input transitions with labels $x^t_k$ and $x^f_k$, where $u_k = v^i_j$,
and the only two states in $P'$ that have outgoing output transitions
with such labels are $q^t_k$ and $q^f_k$,
then $q$ has to be $q^t_k$ or $q^f_k$.

Similarly, the only outgoing transition from $q^D_{i,j}$ in $E$,
if $u_k = v^i_j$, is an output transition with label $x^D_k$,
and the only state in $P'$ that can have an outgoing input transition
with label $x^D_k$ is $q^P_0$.
Therefore, if $(q^D_{i,j}, q)$ is reachable in $E\| P'$ then $q$ is $q^P_0$.

Combining the last two observations,
and because $q^V_{i,j}$ always follows $q^D_{i,j}$ in $E$,
if a state $(q^V_{i,j}, q)$ is reachable, then either
$(q^P_0, x^D_k, q^t_k) \in T$ or $(q^P_0, x^D_k, q^f_k) \in T$.

We construct a truth assignment $t$ for $U$ as follows:
for every $u_k$ such that there exists $i$ and $j$ such that $v^i_j = u_k$
and $(q^V_{i,j}, q)$ is reachable,
we set $t(u_k)$ to $\true$ if $(q^P_0, x^D_k, q^t_k) \in T$
and to $\false$ if $(q^P_0, x^D_k, q^f_k) \in T$.
We assign an arbitrary value to the remaining variables in $U$.
$t$ is well-defined because $P'$ is deterministic and thus it cannot be the case that both $(q^P_0, x^D_k, q^t_k) \in T$ and $(q^P_0, x^D_k, q^f_k) \in T$.

We show that $t$ satisfies all clauses in $C$.

No deadlock states are reachable in $E\| P'$, therefore $(deadlock, q)$ is not reachable.

Therefore for each $j\in [1,m]$ there is $i \in [1,3]$
such that one of the following holds:
\begin{itemize}
\item $z^i_j = u_k$ and $(q^V_{i,j}, q^t_k)$ is reachable or
\item $z^i_j = \overline{u_k}$ and $(q^V_{i,j}, q^f_k)$ is reachable.
\end{itemize}

This guarantees that whenever $(q^D_{1,j}, q^P_0)$ is reachable in a run of $E\| P'$, the run continues to $(q^D_{1,j+1}, q^P_0)$ or $(success, q^P_0)$ and does not reach $(deadlock, q^P_0)$.

In the first case, it has to be that $(q^P_0, x^D_k, q^t_k) \in T$ or $t(u_k) = \true$.

In the second case, it has to be that $(q^P_0, x^D_k, q^f_k) \in T$ or $t(u_k) = \false$.

Hence, for every $j \in [1,m]$ there is $i \in [1,3]$ such that
if $z^i_j = u_k$ then $t(u_k) = \true$
and if $z^i_j = \overline{u_k}$ then $t(u_k) = \false$,
or $t$ satisfies all clauses in $C$.
\end{proof}

\begin{claim}
Problem~\ref{probgeneral} is PSPACE-complete.
\end{claim}
\begin{proof}
Problem~\ref{probgeneral} is in NPSPACE since we can guess a set of
completions, and then check whether the requirements of the problem
are satisfied. Checking whether the
resulting process automata are deterministic can be done in polynomial time.
Checking whether the product automaton is safe, empty, and deadlock-free
can all be done in PSPACE. Since NPSPACE = PSPACE, Problem~\ref{probgeneral} is in PSPACE.

PSPACE-hardness can be seen by observing that a special case of
Problem~\ref{probgeneral} is when the set of process automata is empty.
This special case of determining the safety of a set of environment automata
is at least as hard as reachability, which is PSPACE-hard.
\end{proof}